\pdfoutput=1
\documentclass{article}

\usepackage{PRIMEarxiv}

\usepackage[utf8]{inputenc} 
\usepackage[T1]{fontenc}    
\usepackage{hyperref}       
\usepackage{url}            
\usepackage{booktabs}       
\usepackage{amsfonts}       
\usepackage{nicefrac}       
\usepackage{microtype}      
\usepackage{lipsum}
\usepackage{fancyhdr}       
\usepackage{graphicx}       
\graphicspath{{media/}}     

\usepackage{graphicx}
\usepackage{tikz-cd}
\usepackage[utf8]{inputenc}  
\usepackage{amssymb, amsmath}
\usepackage{amsthm}
\usepackage{hyperref}

\theoremstyle{remark}
\newtheorem{remark}{Remark}[section]

\theoremstyle{definition}
\newtheorem{theorem}{Theorem}[section]

\newtheorem{proposition}{Proposition}[section]
\newtheorem{definition}{Definition}[section]
\newtheorem{example}{Example}[section]

\newcommand{\cinf}{$\mathcal{C}^{\infty}$}

\title{$\mathcal{C}^{\infty}$-symmetries of distributions and integrability
\thanks{\textit{\underline{Citation}}: 
\textbf{Authors. Title. Pages.... DOI:000000/11111.}} 
}

\author{
  *A.J. Pan-Collantes, A. Ruiz, C. Muriel, J.L. Romero \\
  Departamento de Matem\'aticas\\ 
  Universidad de C\'{a}diz - UCA \\
  Puerto Real\\
  \texttt{\{antonio.pan@uca.es, adrian.ruiz@uca.es,concepcion.muriel@uca.es,juanluis.romero@uca.es\}} \\
}

\begin{document}
\maketitle

\begin{abstract}
An extension of the notion of solvable structure for involutive distributions of vector fields is introduced. The new structures are based on a generalization of the concept of symmetry of a distribution of vector fields, inspired in the extension of Lie point symmetries to $\mathcal{C}^{\infty}$-symmetries for ODEs developed in the recent years.  
These new objects, named \cinf-structures, play a fundamental role in the integrability of the distribution: the knowledge of  a \cinf-structure for a corank $k$ involutive distribution permits  to find its integral manifolds by solving  $k$ successive completely integrable Pfaffian equations. These results have important consequences for the integrability of differential equations. In particular,  we derive a new procedure to integrate an $m$th-order ordinary differential equation by splitting the problem into $m$ completely integrable Pfaffian equations. This step-by-step integration procedure is applied to integrate completely several equations that cannot be solved by standard procedures.
\end{abstract}


\section{Introduction}
\textbf{Note}: The published version of this preprint is in  \href{https://www.sciencedirect.com/science/article/pii/S0022039622006982}{Journal of Differential Equations
Volume 348, 5 March 2023, Pages 126-153}.

 It is well known that for a first-order ordinary differential equation (ODE) Lie symmetry  method permits integration by quadrature when a symmetry is known \cite{blumanlibro,olver86,stephani,ibragimovlibro}. Indeed, the knowledge of a Lie point symmetry for a first-order ODE is equivalent to finding a first integral and its corresponding integrating factor. For $m$th-order equations, a one-parameter group of Lie point symmetries permits to reduce the order by one; once the reduced equation has been solved, the solution of the original equation can be constructed by a quadrature. Similarly, a $r$-parameter group of Lie point symmetries with a solvable Lie algebra reduces to a $(m-r)$th-order ODE, and the general solution of the equation can be recovered by $r$ quadratures from the solution of the reduced equation. When $m=r$, the equation can completely be solved after  $m$ quadratures.

However, the existence of a solvable group of Lie point symmetries is not a necessary condition for the integrability by quadratures. Indeed, there exist many examples of ODEs with a trivial group of Lie point symmetries  that can be integrated \cite{olver86,artemioecu,blumanreid,muriel03lie}. This suggests that there must exist objects beyond Lie point symmetries that explain such integrability processes. Among them, solvable structures and $\mathcal{C}^{\infty}$-symmetries have played an important role in problems of integrability of ODEs in the recent literature.

In 1991 Basarab-Horwath characterized the integrability by quadratures by using the concepts of symmetry and solvable structure \cite{basarab,hartl1994solvable} for involutive distributions of vector fields. When applied to the integration of a given ODE, the procedure requires starting with a symmetry $X$ of the distribution spanned by the vector field $A$ associated to the  ODE, i.e., $X$ must satisfy the condition \begin{equation}\label{corchete1}
    [X,A]=\rho A
\end{equation} for some smooth function $\rho.$   The procedure pursues with a symmetry of the distribution generated by $\{A,X\}$ 
and continues  by using an additional symmetry of the subsequent distributions in each stage. In this context, the vector fields $X$ satisfying \eqref{corchete1} correspond to  generalized symmetries \cite{olver86} of the given ODE, also called dynamical symmetries \cite{stephani} or higher-order symmetries \cite{blumanlibro}.

On the other hand, the concept of $\mathcal{C}^{\infty}$-symmetry (or $\lambda$-symmetry) was introduced in \cite{muriel01ima1}, as a  generalization of the notion of Lie point symmetry.  That  concept is based on a new method of prolonging a vector field by using a smooth function $\lambda.$  The resulting $\lambda$-prolonged vector field $X$  satisfies a weaker condition than that of \eqref{corchete1}:  \begin{equation}\label{corchete2}
    [X,A]=\lambda X+\rho A.
\end{equation} Only when $\lambda=0,$ $X$ is a symmetry of the distribution generated by $A$.  Like Lie point symmetries,  $\mathcal{C}^{\infty}$-symmetries can be used to reduce the order or to find first integrals of ODEs.  In this way there are more chances to achieve order reductions or exact solutions, even in the absence of Lie point symmetries  \cite{olver86,artemioecu,blumanreid,muriel03lie}. 

Inspired by these two approaches, in Section \ref{mainresults}  we firstly introduce the concept of \textit{$\mathcal{C}^{\infty}$-symmetry}  of an involutive distribution  of vector fields  $\mathcal{Z}$.
Secondly, the concept of solvable structure is extended by using $\mathcal{C}^{\infty}$-symmetries of distributions instead of symmetries. In this way we achieve a generalization of the notion of solvable structure, which will be called a \textit{structure of $\mathcal{C}^{\infty}$-symmetries}  (or simply a \textit{$\mathcal{C}^{\infty}$-structure}) for $\mathcal{Z}$. We also provide the concept of \textit{$\mathcal{C}^{\infty}$-structure of 1-forms} for $\mathcal{Z},$ establishing the correspondence between these new two concepts (Proposition \ref{qsolv_struct_dual}). After that, we investigate how the main result in \cite{basarab}, concerning the integrability by quadratures through solvable structures, can be generalized when $\mathcal{C}^{\infty}$-structures are considered. It is proven that the knowledge of a $\mathcal{C}^{\infty}$-structure for a corank $k$ involutive distribution permits to find the integral manifolds by solving $k$ completely integrable Pfaffian equations (Theorem \ref{maintheorem}).

Since the elements of a \textit{$\mathcal{C}^{\infty}$-structure} satisfy weaker conditions than those required for a solvable structure, the new approach enlarges the class of vector fields that can be used to integrate involutive distributions. On the other hand,  it must be said that a drawback is that the  Pfaffian equations associated to a   \textit{$\mathcal{C}^{\infty}$-structure}, although completely integrable,  might not be defined by closed 1-forms, as in the case of solvable structures. 

These results are firstly established in the  general context of distributions and later applied in Section \ref{applications} to integrability problems of differential equations. In particular, we prove that for a $m$th-order ODE, the knowledge of a $\mathcal{C}^{\infty}$-structure permits to {\it split} the ODE into $m$ completely integrable Pfaffian equations.

The new approach, apart from its theoretical interest, broadens the  class of vector fields that can be involved in the process of integrating differential equations, enlarging the range of problems that can be addressed.  Indeed, as we will show in Section \ref{applications}, a $\mathcal{C}^{\infty}$-structure can be used to integrate ODEs that present some obstacle when are attempted to be  solved by using standard procedures.

The examples in this section
aim to show how \cinf-structures can be constructed and applied to solve this type of problems. The equation studied in Example \ref{example1lie} admits just one Lie point symmetry. The Lie symmetry method fails because  the quadrature that recovers the solutions from the reduced equation (which is of Abel type) cannot be explicitly evaluated. For the equations in examples \ref{exampleord2sinlie} and  \ref{exampleorden3} the Lie symmetry method cannot be applied because the equations lack Lie point symmetries. All of them can be completely integrated by using  the $\mathcal{C}^{\infty}$-structure approach. Finally, we consider in Example \ref{examplesystem} a non-autonomous first-order system of two ordinary differential equations to illustrate how the  method can also be used for systems of ODEs.

\section{Preliminaries}\label{preliminaries}

In this section we introduce some notation and theoretical background of our work.

\subsection{Basic assumptions and notations}

Due to the local character of our results, we will work in an open simply connected 
 subset $U$ of an Euclidean space $\mathbb{R}^n$.  All the results can be extended straightforwardly to the case of an $n$-dimensional manifold. The vector fields, differential forms and functions are assumed to be smooth (meaning $\mathcal{C}^{\infty}$) although this condition could be relaxed for each particular situation.  We also assume that the domains have been restricted when necessary.

By $\mathfrak{X}(U)$ and ${\Omega}^k(U)$ we will denote the modules over $\mathcal{C}^{\infty}(U)$ of all smooth vector fields and $k$-forms, respectively, while $\Omega^*(U)$ will denote the exterior algebra of all the differential forms on $U$ \cite{Morita,warner}.

Given $r$ pointwise linearly independent vector fields $Z_1,\ldots,Z_r \in \mathfrak{X}(U),$ we will denote by
\begin{equation}\label{distribucioncampos}
  \mathcal{Z}:=\mathcal{S}(\{Z_1,\ldots,Z_r\})  
\end{equation}
the submodule of $\mathfrak{X}(U)$ generated by $Z_1,\ldots,Z_r.$ Similarly, the submodule of $\Omega^1(U)$ generated by a set of $r$ pointwise linearly independent 1-forms $\Lambda=\{\omega_1,\ldots,\omega_{r}\}$ will be denoted by $ \mathcal{S}(\Lambda).  $ 

The submodule  $\mathcal{Z}$ (resp.  $ \mathcal{S}(\Lambda) $) defines a distribution (resp. a Pfaffian differential system) of constant rank $r$.  

A $k$-form $\omega$ is said to annihilate $\mathcal{Z}$ if $\omega(Y_1,\ldots,Y_k)=0$ on $U$ whenever $Y_1,\dots,Y_k\in \mathcal{Z}$ \cite[Definition 2.26]{warner}; we denote
\begin{equation}
    \mbox{Ann}(\mathcal{Z})=\{\omega\in \Omega^*(U): \omega \quad \mbox{annihilates}\quad  \mathcal{Z}\}.
\end{equation} The set   $\mbox{Ann}(\mathcal{Z})$ is  an ideal in $\Omega^*(U)$ which  is locally generated by $n-r$ pointwise linearly independent 1-forms $\Lambda:=\{\omega_1,\ldots,\omega_{n-r}\}$  \cite[Proposition 2.28]{warner}. In this case we will write
\begin{equation}\label{anuladorideal}
    \mbox{Ann}(\mathcal{Z})=\mathcal{I}(\Lambda).
\end{equation}
The associated submodule $\mathcal{S}(\Lambda)$ will be denoted by $\mathcal{Z}^{\circ}$ and it can be easily checked that 
\begin{equation}\label{Z*}
    \mathcal{Z}^{\circ}=\mathcal{S}(\Lambda)=\{\omega\in \Omega^1(U): Z\,\lrcorner\,\omega=0, \,\mbox{for each} \,Z \in \mathcal{Z}\},
\end{equation}
where 
$\left. \right\lrcorner$ denotes contraction (interior product). By linearity, it is clear that $\omega\in \mathcal{Z}^{\circ}$ (resp. $Z\in \mathcal{Z}$) if and only if  $Z_k \,\lrcorner\,\omega =0$ for $1\leq k\leq r$ (resp.  $Z \,\lrcorner\, \omega_l =0$ for $1\leq l\leq n-r$).

The distribution  $\mathcal{Z}$ is said to be involutive if it is closed under the Lie bracket.  This involutivity can be characterized in terms of   the algebraic ideal \eqref{anuladorideal}. For the following proposition, let us  recall that an ideal $\mathcal{J}\subset \Omega^*(U)$ is called a differential ideal \cite[Definition 2.29]{warner} if it is is closed under exterior differentiation $d;$ that is 
$d(\mathcal{J})\subset \mathcal{J}.$
\begin{proposition}\label{caracterizacioninvolutivo}\cite[Proposition 2.30]{warner}
A distribution  $\mathcal{Z}$ is involutive if and only if  the ideal $\mbox{Ann}(\mathcal{Z})=\mathcal{I}(\Lambda)$ is a differential ideal. 
\end{proposition}

A smooth function $F$ on $U$ is a first integral of  $\mathcal{Z}$ if $Z(F)=0$ for each $Z\in \mathcal{Z}$ or, equivalently, if $Z_i(F)=0$ for $1\leq i\leq r.$ Similarly, we say that $F$ is  a first integral of $\mathcal{S}(\Lambda)$ if $dF\in \mathcal{S}(\Lambda)$ or, equivalently, if
\begin{equation}\label{edpintegralprimera}
    dF\wedge\omega_1\wedge\ldots\wedge \omega_{n-r}=0.
\end{equation} 
Therefore, a function $F$ is a first integral of  $\mathcal{Z}$ if and only if   $F$ is a first integral of $\mathcal{Z}^{\circ}.$

\subsection{Distributions and solvable structures}

For an involutive distribution \eqref{distribucioncampos}, Frobenius Theorem establishes the local existence of a complete set of first integrals of $\mathcal{Z},$ i.e., $n-r$ functionally independent first integrals of $\mathcal{Z}.$ The  integral manifolds{\footnote{In this paper integral manifolds of distributions will always be integral manifolds of maximal dimension \cite[Remark 1.58]{warner}.}} of $\mathcal{Z}$   can be implicitly defined by the corresponding level sets. In fact,  Frobenius Theorem  guarantees, for $p\in U,$ the existence of a unique maximal \cite[Definition 1.63]{warner} connected integral manifold of the involutive distribution through $p$ \cite[Theorem 1.64]{warner}. However, Frobenius Theorem 
does not  provide  a procedure to find the integral manifolds.

An important notion concerning the search for integral manifolds of distributions is that of solvable structure, which is based on the following concept of symmetry of a distribution \cite{basarab,hartl1994solvable,sherring1992geometric,Barco2001,Barco2002,Catalano_Paola,lychagin1991}:  

\begin{definition}\label{simetriadedistribucion}
Let $\mathcal{Z}=\mathcal{S}(\{ Z_1,\ldots, Z_r \}) $ be an involutive distribution of vector fields on an open set $U\subset\mathbb{R}^n.$ A 
vector field $X$ on  $U$ is called a symmetry of  $\mathcal{Z}$ if :
\begin{enumerate}
\item $ \{Z_1,\ldots, Z_r, X\} $ is a set of pointwise linearly independent vector fields on $U;$ and
\item there exist smooth functions $c_{ik}$ on $U$ such that
$$
[X,Z_i]=\sum_{k=1}^r c_{ik} Z_k \text{, } \text{ for } i=1, \ldots,r.
$$
\end{enumerate}
\end{definition}

The previous notion of symmetry has been used in \cite{basarab} to define a solvable structure as follows: 

\begin{definition}\label{solvable}
Let $ \mathcal{Z}=\mathcal{S}(\{Z_1,\ldots,Z_r\}) $ be an involutive distribution on an open subset $U$ of     
 $ \mathbb{R}^n$. Let $ \langle X_1,\ldots, X_{n-r}\rangle $ be an ordered collection of vector fields on $U$ and denote 
\begin{equation}\label{XK}
\begin{array}{l}
    \mathcal{X}_0:=\{ Z_1,\ldots,Z_r\}\quad\mbox{and}\quad
    \mathcal{X}_k:=\{ Z_1,\ldots,Z_r, X_1,\ldots,X_k\}, \quad\mbox{for}\quad 1\leq k\leq n-r.
    \end{array}
\end{equation}
We will say that $ \langle X_1,\ldots, X_{n-r}\rangle $ is a solvable structure  for $ \mathcal{Z} $ if  for $1\leq k\leq n-r$ the vector field $ X_k $ is a symmetry of the rank $r+k-1$ distribution $ \mathcal{S}(\mathcal{X}_{k-1}).$ 
\end{definition}

\begin{remark}

\begin{enumerate}
    \item[(a)] In references such as \cite{basarab,hartl1994solvable,sherring1992geometric,Barco2001,Barco2002,Catalano_Paola} the term solvable structure refers to the whole ordered set $ \langle Z_1,\ldots,Z_r, X_1,\ldots, X_{n-r}\rangle,$ while in this work  the term solvable structure for $\mathcal{Z}$ refers only to the part  $\langle X_1,\ldots, X_{n-r}\rangle,$ because the distribution $\mathcal{Z}$ is fixed and given in advance. 
    \item [(b)] Observe that in Definition \ref{solvable}, for $0\leq k \leq n-r,$ the vector fields in $\mathcal{X}_k$ are pointwise linearly independent on $U.$ In particular,  $X_1,\ldots,X_{n-r}$ are also pointwise linearly independent.
\end{enumerate}
\end{remark}

The knowledge of a solvable structure provides a systematic procedure to find the integral manifolds of  the distribution $ \mathcal{Z} $ by   integrating successively certain 1-forms  constructed from  the elements of the solvable structure. At each stage there appears a 1-form that can be integrated by quadrature.    This is the main result on solvable structures concerning the integrability of the distribution (see \cite{basarab,hartl1994solvable,sherring1992geometric,Barco2001,Barco2002,Catalano_Paola} for further details): 

\begin{proposition} \cite[Proposition 3]{basarab}\label{basarabprop3}
Given a solvable structure for an involutive distribution $ \mathcal{Z} $, their integral manifolds can be found, at least locally, by quadratures alone.
\end{proposition}

\begin{remark}\label{remarkPfaffian}
Here the term quadrature refers to the problem of finding a smooth function $F$ such that locally $dF=\omega$, for a previously given exact 1-form $\omega$. Such function $F$ is called a primitive  of  $\omega.$ Since \eqref{edpintegralprimera} is satisfied, $F$ is a first integral  of  $\mathcal{S}(\{\omega\}).$  The level sets of $F$ are integral manifolds of the distribution annihilated by $\mathcal{S}(\{\omega\}).$ 

For further reference, let us recall that if $\omega$ is not exact but at least $\omega \wedge d\omega=0$ then $\mathcal{I}(\{\omega\})$ is a differential ideal and the distribution associated to the Pfaffian system $\mathcal{S}(\{\omega\})$ is involutive. By Frobenius Theorem there also exist (at least locally) integral manifolds described  by the level sets of a smooth function $F$. Therefore there also exists  a non-vanishing smooth function $\mu$, such that
$dF=\mu \omega$.
In this case it is said that $\omega \equiv 0
$ is a \emph{completely integrable Pfaffian equation} or that the 1-form $\omega$ is \emph{Frobenius integrable}. 
The function $\mu$ is called integrating factor of $\omega$ and $F$ is a first integral of $\mathcal{S}(\{\omega\})$. 
\end{remark}

\subsection{Symmetries and solvable structures for ODEs }\label{symmetries}

In this section we adapt the former concepts and results to the framework of the integrability of a given $m$th-order
ODE 
\begin{equation}\label{ODE1}
u_m=\phi(x, u,u_1,\ldots,u_{m-1}),
\end{equation}
where $(x,u)$ are in some open set $M\subset \mathbb{R}^2$ and represent the independent and dependent variables, respectively, and $u_i$  denotes the derivative of $u$ with respect to $x$ of order $i$, for $i=1,\dots,m.$  The corresponding $m$th-order jet space \cite{saunders1989geometry} is denoted by $J^m(\mathbb{R},\mathbb{R})$ and $\phi$ is a smooth function defined on the open subset $M^{(m-1)}\subseteq J^{m-1}(\mathbb{R},\mathbb{R}).$

Equation \eqref{ODE1} determines the rank 1 distribution $\mathcal{S}(\{A\})$ defined by the vector field
\begin{equation}\label{campoA}
A=\partial_x+ u_1 \partial_u+\ldots+ \phi \partial_{u_{m-1}}\in \mathfrak{X}(M^{(m-1)}).
\end{equation}

The distribution $\mathcal{S}(\{A\})$ is trivially involutive and hence integrable. Its integral manifolds (curves) correspond to solutions of  (\ref{ODE1}) prolonged to the $(m-1)$th-order jet space.

The Definition \ref{simetriadedistribucion} of symmetry of a distribution, when applied to $\mathcal{S}(\{A\})$, includes well-known notions of symmetry of an ODE. Given a  Lie point symmetry of (\ref{ODE1})
\begin{equation}\label{pointsymmetry}
\mathbf{v}=\xi(x,u) \partial_x +\eta (x,u) \partial_u,
\end{equation}
let  $X:=\mathbf{v}^{(m-1)}$ denote the standard $(m-1)$th-order prolongation of $\mathbf{v}$ \cite{olver86}. The vector field $X$ satisfies the relation \cite{stephani}:
\begin{equation}\label{bracketsymmetry}
[X,A]= \rho A,  
\end{equation}
where $\rho=-A(\xi).$ Therefore, according to Definition \ref{simetriadedistribucion}, $X$ is a symmetry of the distribution $\mathcal{S}(\{A\}).$  When the functions $\xi$ and $\eta$ in \eqref{pointsymmetry} also depend on 
derivatives of $u$ we obtain \emph{generalized} symmetries  \cite{olver86} (also called \emph{dynamical} symmetries  \cite{stephani} or \emph{higher-order} symmetries \cite{blumanlibro}). They satisfy the same condition (\ref{bracketsymmetry})  \cite[Chapter 12]{stephani} and are therefore symmetries in the sense of distributions. 

In this context, a solvable structure for $\mathcal{S}(\{A\})$ consists of an ordered collection of vector fields $\langle X_1,\ldots, X_{m} \rangle$ such that:
\begin{enumerate}
    \item[(i)] $X_1$  is a generalized symmetry of (\ref{ODE1}). 
        \item[(ii)] For $2\leq k \leq m$ the vector field $X_k$ is a symmetry, in the sense of Definition \ref{simetriadedistribucion}, of the distribution $\mathcal{S}(\{A,X_1,\ldots,X_{k-1}\})$.
\end{enumerate}

It must be observed that in general the vector fields  $X_k$ for $2\leq k\leq m$ are not necessarily generalized symmetries of the equation. On the other hand, the $m$ vector fields of a solvable structure do not necessarily constitute a solvable Lie algebra. However, a $m$-dimensional solvable algebra of symmetries provides a particular case of solvable structure \cite{basarab,hartl1994solvable,sherring1992geometric,Barco2001,Barco2002}.  

As stated in Proposition \ref{basarabprop3}, the knowledge of a solvable structure is sufficient to integrate the equation by quadratures. The reader is referred to \cite{basarab,hartl1994solvable,sherring1992geometric,Barco2001,Barco2002,Catalano_Paola}
for more details on the procedure.

\section{$\mathcal{C}^{\infty}$-structure for a distribution}\label{mainresults}
 
Although most of the classical procedures to integrate a given ODE are consequence of the existence of  Lie point symmetries, there are many ODEs without Lie point symmetries that can be integrated. This motivated the search of a  wider class of vector fields that could explain such integration procedures. In this context,  $\mathcal{C}^{\infty}$-symmetries (or $\lambda$-symmetries) \cite{muriel01ima1} and its posterior generalizations \cite{ muriel2017IMA,gaetamorando,murielolver,cicogna2012generalization,muriel_evolution,gaetatwisted,gaeta2,zhang2019relationship}
emerged as  extensions of the Lie point symmetry concept.  
The notion of  $\mathcal{C}^{\infty}$-symmetry is based on a different way of prolonging a given vector field \eqref{pointsymmetry}, which uses a smooth function $\lambda.$ If a pair $(\mathbf{v},\lambda)$ is a $\mathcal{C}^{\infty}$-symmetry of the equation \eqref{ODE1} then the  $(m-1)$th-order $\lambda$-prolongation of $\mathbf{v}$ denoted by $X:=\mathbf{v}^{[\lambda,(m-1)]}$ \cite{muriel01ima1}, satisfies
\begin{equation}\label{bracketlambdasym}
[X,A]=\lambda X+\rho A,
\end{equation} which provides a  weaker condition than  condition (\ref{bracketsymmetry}).

When  instead of vector fields of the form \eqref{pointsymmetry}, generalized vector fields are considered, the concept of $\mathcal{C}^{\infty}$-symmetry can be extended in the following straightforward way: 
\begin{definition}\label{general_cinf-sym}
A vector field $X\in \mathfrak{X}(M^{(m-1)})$ will be called a generalized $\mathcal{C}^{\infty}$-symmetry of equation (\ref{ODE1}) if there exist $\lambda,\rho \in \mathcal{C}^{\infty}(M^{(m-1)})$ such that
$$
[X,A]=\lambda X+\rho A.
$$
\end{definition}

Our next goal is to use a similar idea in order to introduce the notion of $\mathcal{C}^{\infty}$-symmetry of a distribution.

\begin{definition}\label{Csymmdistribution}
Let $\mathcal{Z}=\mathcal{S}(\{ Z_1,\ldots, Z_r \})$ be an involutive distribution  on $U.$ A vector field $X$ on $U$ is called a $\mathcal{C}^{\infty} $-symmetry of $\mathcal{Z}$ if:
\begin{enumerate}
\item $ \{Z_1,\ldots, Z_r, X \}$ is a set of pointwise linearly independent vector fields on $U;$ and 
\item the distribution $\mathcal{S}(\{Z_1,\ldots, Z_r, X \})$ is involutive, i.e., there exist smooth functions $\lambda_i, c_{ik}$ on $U$ ($ i,k=1,\ldots,r$) such that
\begin{equation}\label{deficinfsym}
[X,Z_i]=\lambda_i X+\sum_{k=1}^r c_{ik} Z_k \text{, } i=1, \ldots,r.
\end{equation}

\end{enumerate}
\end{definition}

\begin{remark}\label{remark_cinfsym}
\begin{enumerate}
    \item[(a)]
The notion of symmetry of a distribution in Definition \ref{simetriadedistribucion} is a particular case of the concept given in Definition \ref{Csymmdistribution}. It corresponds to the particular case when 
    $\lambda_i=0$ for  $i=1,\ldots,r.$ 
      \item[(b)] The notion of $\lambda$-symmetry of a distribution in \cite[Definition 2]{zhang2019relationship} is a particular case of the concept given in Definition \ref{Csymmdistribution}. It corresponds to the particular case when
    $\lambda_i=\lambda_j$ for $i,j=1,\ldots,r.$
    \item[(c)] A $ \mathcal{C}^{\infty} $-symmetry of the distribution $\mathcal{Z}=\mathcal{S}(\{A\})$ generated by the vector field \eqref{campoA} associated to equation \eqref{ODE1} is a generalized $\mathcal{C}^{\infty}$-symmetry of \eqref{ODE1}  in the sense of Definition \ref{general_cinf-sym}.
    \item[(d)] If $\mathcal{Z}=\mathcal{S}(\{Z_1,\ldots,Z_{n-1}\})$ is an involutive distribution of corank 1 then $X$ is a $ \mathcal{C}^{\infty} $-symmetry of $\mathcal{Z}$ if and only if $\{Z_1,\ldots,Z_{n-1},X\}$ is a set of pointwise linearly independent vector fields. 
\end{enumerate}
\end{remark}

The following concept extends the notion of solvable structure (Definition \ref{solvable}) when we use  $\mathcal{C}^{\infty}$-symmetries instead of symmetries:

\begin{definition}\label{Qsolvable}
Let $ \mathcal{Z}=\mathcal{S}(\{Z_1,\ldots,Z_r\}) $ be an involutive distribution on an open subset $U$ of $\mathbb{R}^n$. Let $ \langle X_1,\ldots, X_{n-r}\rangle $ be an ordered collection of vector fields on $U$ and  define the sets $ \mathcal{X}_k,$ for $0\leq k\leq n-r,$   as in \eqref{XK}.
We will say that $ \langle X_1,\ldots, X_{n-r}\rangle $ is a structure of \cinf-symmetries for $\mathcal{Z}$, or simply a \cinf-structure for $\mathcal{Z}$, if for $1\leq k\leq n-r$ the vector field $X_k$ is a $ \mathcal{C}^{\infty}$-symmetry of the rank $r+k-1$ distribution $ \mathcal{S}(\mathcal{X}_{k-1}).$
\end{definition}

\begin{remark}\label{involutive}
\begin{enumerate}
    \item[(a)] 
    Condition 1 in Definition \ref{Csymmdistribution} implies that for $0\leq k\leq n-r$ the sets $\mathcal{X}_k$ appearing in Definition \ref{Qsolvable} (and in particular   $\{X_1,\ldots,X_{n-r}\}$) are  sets of  pointwise linearly independent vector fields on $U.$ 
    \item[(b)] In Definition \ref{Qsolvable} the distributions $\mathcal{S}(\mathcal{X}_k)$ are  necessarily involutive for $k=0,\ldots,n-r.$ In fact,  $ \langle X_1,\ldots, X_{n-r}\rangle $ is a $\mathcal{C}^{\infty}$-structure  for $ \mathcal{Z}$ if and only if  for $1\leq k\leq n-r$ the distribution $\mathcal{S}(\mathcal{X}_k)$ is involutive and  $\mathcal{X}_{n-r}$ is a set of pointwise linearly independent vector fields on $U.$
    \item[(c)] A solvable structure is a particular case of a $\mathcal{C}^{\infty}$-structure,  in which all the vector fields $X_i$ are  symmetries of the corresponding distribution. 
    \item[(d)] The last element of a \cinf-structure can be any vector field $X_{n-r}$ such that $\mathcal{X}_{n-r}$ is  a set of pointwise linearly independent vector fields on $U$ (see Remark \ref{remark_cinfsym} (d)).
\end{enumerate}
\end{remark}

By using Proposition \ref{caracterizacioninvolutivo}, the conditions on the involutivity of $\mathcal{S}(\mathcal{X}_k),$ for $0\leq k\leq n-r,$ (see Remark \ref{involutive} (b)) in Definition \ref{Qsolvable}  can also be expressed in terms of the associated Pfaffian system $\mathcal{S}(\mathcal{X}_k)^{\circ}.$ This motivates the following definition:  
\begin{definition}\label{Qsolvabledual}
Let $ \mathcal{Z}=\mathcal{S}(\{Z_1,\ldots,Z_r\}) $ be an involutive distribution on an open set $U\subset\mathbb{R}^n$. Let $\langle \omega_1,\ldots, \omega_{n-r}\rangle$ be an ordered collection of pointwise linearly independent 1-forms on $U$ and define the sets
\begin{equation}\label{wk}
    \Lambda_k:=\{ \omega_{k+1},\ldots, \omega_{n-r}\}\quad\mbox{for}\quad 0\leq k\leq n-r-1.
\end{equation} 
We will say that $\langle \omega_1,\ldots, \omega_{n-r}\rangle$ is a $\mathcal{C}^{\infty}$-structure of 1-forms for $ \mathcal{Z} $ if the following two conditions hold:
\begin{enumerate}
\item[(a)] $ \mathcal{Z}^{\circ}=\mathcal{S}\left(\Lambda_0\right)$.
\item[(b)]\label{segundacondicion} $ \mathcal{I}(\Lambda_k) $ is a differential ideal for $k=1,\ldots, n-r-1$.
\end{enumerate}
\end{definition}

The following proposition states that there exists a correspondence between the objects introduced in  Definition \ref{Qsolvable} and Definition \ref{Qsolvabledual}:

\begin{proposition}\label{qsolv_struct_dual}
Let $\mathcal{Z}=\mathcal{S}(\{Z_1,\ldots,Z_r\})$ be an involutive distribution on an open subset $U$ of $\mathbb{R}^n$. 
Given a \cinf-structure $\langle X_1,\ldots, X_{n-r}\rangle$  for $\mathcal{Z}$, there exists a \cinf-structure of 1-forms $\langle \omega_1,\ldots, \omega_{n-r}\rangle$ such that
\begin{equation}\label{eq_condiciondual}
    \mathcal{S}\left( \mathcal{X}_{k} \right)^{\circ}=\mathcal{S}\left( \Lambda_{k} \right),\quad\mbox{for}\quad 1\leq k\leq n-r-1  
\end{equation}
and conversely.  When conditions \eqref{eq_condiciondual} are satisfied we will say that $\langle X_1,\ldots, X_{n-r}\rangle$  and   $\langle \omega_1,\ldots, \omega_{n-r}\rangle$ are $\mathcal{Z}$-related. 
\end{proposition}

\begin{proof}
Given a \cinf-structure $\langle X_1,\ldots, X_{n-r}\rangle$ for $\mathcal{Z}$, consider the ordered set $\langle \omega_1,\ldots, \omega_{n-r}\rangle$ of the last $n-r$ elements of the dual frame 
of $\mathcal{X}_{n-r}$ and define the corresponding sets given in \eqref{wk}. By construction, the elements in $\Lambda_0$ are pointwise linearly independent and, by duality, 
$\mathcal{Z}^{\circ}=\mathcal{S}(\Lambda_0)$ and conditions \eqref{eq_condiciondual} hold. Since for $1\leq k\leq n-r-1$ the distribution $\mathcal{S}\left( \mathcal{X}_{k} \right)$ is involutive, then  
$\mathcal{I}(\Lambda_k)$ is a differential ideal (see Proposition \ref{caracterizacioninvolutivo}). This proves that $\langle \omega_1,\ldots, \omega_{n-r}\rangle$ is a \cinf-structure of 1-forms for $\mathcal{Z}$ satisfying conditions \eqref{eq_condiciondual}.

Conversely, given a \cinf-structure of 1-forms $\langle \omega_1,\ldots, \omega_{n-r}\rangle$ for  $\mathcal{Z},$ let $\{\zeta_1,\ldots,\zeta_r\}$ be a set of 1-forms that completes $\Lambda_0$ to a coframe  $\theta:=\{\zeta_1,\ldots,\zeta_r, \omega_1,\ldots, \omega_{n-r}\}$ on $U.$ Consider the ordered set $\langle X_1,\ldots, X_{n-r}\rangle$ of the last $n-r$ elements of the dual basis of $\theta$ and define the corresponding sets given in \eqref{XK}. By construction 
\begin{equation}\label{eq_vectoresduales}
X_i \,\lrcorner\,\omega_j= 0    \text{ on } U, \text{ for } 1\leq i,j \leq n-r, \, i \neq j.
\end{equation}
 Since $\mathcal{Z}^{\circ}=\mathcal{S}(\Lambda_0),$ it follows from \eqref{eq_vectoresduales} that conditions \eqref{eq_condiciondual} hold.
 
In order to prove that $\mathcal{X}_{n-r}$ is a set of pointwise linearly independent vector fields on $U,$ observe that if there exists $p\in U$ such that
 \begin{equation}
     \alpha_1({Z_1})_{p}+\ldots+\alpha_r({Z_r})_{p}+\alpha_{r+1} ({X_1})_p+\ldots+\alpha_{n} ({X_{n-r}})_p=0
 \end{equation}
for some constants $\alpha_i\in \mathbb{R},$ then after  contraction  with $({\omega_i})_p,$ 
we obtain that $\alpha_{r+i}=0,$ for $1\leq i\leq n-r.$ Since $Z_1,\ldots, Z_r$ are pointwise linearly independent vector fields on $U,$ this implies that also $\alpha_{i}=0,$ for $1\leq i\leq r.$ 
 
 Finally, the involutivity of the distributions $\mathcal{S}\left( \mathcal{X}_{k} \right)$ follows from Proposition \ref{caracterizacioninvolutivo}.  
\end{proof}

\begin{remark}\label{weakercondition}
 \begin{enumerate}

\item[(a)]   Observe that the elements of two $\mathcal{Z}$-related \cinf-structures do not necessarily satisfy $X_i\,\lrcorner\,\omega_j=\delta_{ij},$ for $1\leq i,j\leq n-r.$ Conditions \eqref{eq_condiciondual} only imply that $X_i\,\lrcorner\,\omega_j=0$ for $1\leq i<j\leq n-r.$  Besides,  since  we have that $\omega_{i}\in\mathcal{S}(\mathcal{X}_{i-1})^{\circ}$ and $\omega_{i}\notin\mathcal{S}(\mathcal{X}_{i})^{\circ},$ then it can be checked that, for $1\leq i\leq n-r,$ the function $X_{i}\,\lrcorner\,\omega_{i}$ does not vanish on $U.$ 

\item[(b)]   If $\langle X_1,\ldots,X_{n-r}\rangle$ and $\langle \omega_1,\ldots,\omega_{n-r}\rangle$ are 
$\mathcal{Z}$-related and  $h$ is any non-vanishing function on $U,$ then it can be checked that  $\langle X_1,\ldots,hX_{i_0},\ldots,X_{n-r}\rangle$ and $\langle \omega_1,\ldots,\omega_{n-r}\rangle$ are also 
$\mathcal{Z}$-related.
 As a consequence, since the functions 
 $h_{i}=1/(X_{{i}}\,\lrcorner\,\omega_{{i}})$ for $1\leq i\leq n-r$ are well defined on $U,$ then we have that $\langle h_1X_1,\ldots,h_{n-r}X_{n-r}\rangle$ and $\langle \omega_1,\ldots,\omega_{n-r}\rangle$ are  
 $\mathcal{Z}$-related and satisfy the conditions $(h_{i}X_{i})\,\lrcorner\,\omega_{i}=1$ for $1\leq i\leq n-r.$  
 
   \item[(c)] In the context of Definition \ref{Qsolvabledual}, the condition $d(\mathcal{I}(\Lambda_k)) \subseteq \mathcal{I}(\Lambda_{k})$  is satisfied for $0\leq k\leq  n-r.$ 
    More explicitly,  for $1\leq i \leq n-r,$ there exist 1-forms $\sigma_i$ such that
\begin{equation}\label{diferenciales}
    d\omega_i=\sigma_i\wedge \omega_i+\sum_{j=i+1}^{n-r} \sigma_{ij}\wedge \omega_j, 
\end{equation} 
for some 1-forms $\sigma_{ij},$ $j=i+1,\ldots, n-r.$

The notion of solvable structure (Definition \ref{solvable}) can also be formulated in terms of differential forms \cite{basarab,hartl1994solvable}. In this case the corresponding conditions become   $d(\mathcal{I}(\Lambda_k)) \subseteq \mathcal{I}(\Lambda_{k+1})$ for $k=0,\ldots, n-r-2,$ and $d(\mathcal{I}(\Lambda_{n-r-1}))=0,$ which are more  restrictive than in the case of $\mathcal{C}^{\infty}$-structures. In other words, for a solvable structure, conditions (\ref{diferenciales}) become 
\begin{equation}\label{diferencialessolvable}
    d\omega_i=\sum_{j=i+1}^{n-r} \sigma_{ij}\wedge \omega_j. 
\end{equation}

This discussion explains, in terms of 1-forms, why  $\mathcal{C}^{\infty}$-structures require weaker conditions than solvable structures to be determined.
\end{enumerate}
\end{remark}

The following theorem states the main result on  the application of $\mathcal{C}^{\infty}$-structures to the integration of involutive distributions: 

\begin{theorem}\label{maintheorem}
Let $\mathcal{Z}=\mathcal{S}(\{Z_1,\ldots,Z_r\}) $ be  an involutive distribution of vector fields on $U\subset \mathbb{R}^n.$ Any  $\mathcal{C}^{\infty}$-structure 
for $ \mathcal{Z} $ can be used to find the integral manifolds of $ \mathcal{Z} $ by solving  successively $n-r $ completely integrable Pfaffian equations.
\end{theorem}

\begin{proof}
By Proposition \ref{qsolv_struct_dual} the $\mathcal{C}^{\infty}$-structure for $\mathcal{Z}$ can be given in terms of vector fields $\langle X_1,\ldots,X_{n-r}\rangle$ or by a $\mathcal{C}^{\infty}$-structure of 1-forms  $\langle \omega_1,\ldots, \omega_{n-r}\rangle$ satisfying conditions  \eqref{eq_condiciondual} (i.e., $\mathcal{Z}$-related).

According to condition (b) in Definition \ref{Qsolvabledual}, the ideal $\mathcal{I}(\Lambda_{n-r-1})$ is a differential ideal. Therefore  $
\omega_{n-r}\wedge d\omega_{n-r}=0
$
and hence the Pfaffian equation $$\omega_{n-r}\equiv 0$$ is completely integrable. Given $q\in U,$ let  $\Sigma_{n-r}$ be a $(n-1)$-dimensional integral manifold   of $\mathcal{S}(\{\omega_{n-r}\})$  passing through $q$ \cite[Theorem 2.32]{warner}. 
Consider a local parametrization of $\Sigma_{n-r}$ defined by a smooth imbedding \cite[Definition 1.27]{warner}
$\iota:\bar{U}
\to U$ on some open set $\bar{U}
\subset \mathbb{R}^{n-1}.$ 
We define
\begin{equation}
   \bar{\omega}_{i}:=\iota^*(\omega_{i})\, {\text{ for }} 1\leq i\leq n-r. 
\end{equation}  
Since  $\Sigma_{n-r}$ is an integral manifold of  $\mathcal{S}(\{\omega_{n-r}\})$ then $\bar{\omega}_{n-r}=\iota^*(\omega_{n-r})=0.$ On the other hand, since $\mathcal{S}(\{\omega_{n-r}\})=\mathcal{S}(\mathcal{X}_{n-r-1})^{\circ},$  the vector fields \begin{equation}
\begin{array}{ll}
\bar{Z_i}:=(\iota^{-1})_*\left({Z_{i}|}_{{\Sigma}_{n-r}}\right) & \, {\text{ for }} 1\leq i\leq r, \\
\bar{X_j}:=(\iota^{-1})_*\left({X_{j}|}_{{\Sigma}_{n-r}}\right) & \, {\text{ for }} 1\leq j\leq n-r-1
    \end{array}
\end{equation} are well defined.  It follows from \eqref{diferenciales}  and $\bar{\omega}_{n-r}=0$ that, for $1\leq i\leq n-r-1,$
$$d\bar{\omega}_i=\iota^*(d\omega_i)=\bar{\sigma}_i\wedge\bar{\omega}_i+\displaystyle\sum_{j=i+1}^{n-r-1}\bar{\sigma}_{ij}\wedge\bar{\omega}_j,$$ where we have written $\bar{\sigma}_i:=\iota^*(\sigma_i),$ and $\bar{\sigma}_{ij}:=\iota^*(\sigma_{ij}).$
This proves that  $\langle \bar{\omega}_1,\ldots, \bar{\omega}_{n-r-1}\rangle$ is a \cinf-structure of 1-forms for $\bar{\mathcal{Z}}=\mathcal{S}(\bar{Z}_1,\ldots,\bar{Z}_{r})$ on $\bar{U}
\subset \mathbb{R}^{n-1}.$ It can be checked that, by construction,  $\langle \bar{X}_1,\ldots, \bar{X}_{n-r-1}\rangle$ is a \cinf-structure of  $\bar{\mathcal{Z}}$ and that it is $\bar{\mathcal{Z}}$-related  to $\langle \bar{\omega}_1,\ldots, \bar{\omega}_{n-r-1}\rangle.$

In consequence, $\mathcal{I}(\{\bar{\omega}_{n-r-1}\})$ is a differential ideal; hence the Pfaffian equation $$\bar{\omega}_{n-r-1}\equiv 0$$ is completely integrable.
As before, let $\bar{\Sigma}_{n-r-1}$  denote an  $(n-2)$-dimensional integral manifold of $\mathcal{S}(\{\bar{\omega}_{n-r-1}\})$  passing through $\bar{q}=\iota^{-1}(q).$ 
By construction, 
${\Sigma}_{n-r-1}:=\iota(\bar{\Sigma}_{n-r-1})$ is a submanifold of  ${\Sigma}_{n-r}$ and it is an integral manifold of  $\mathcal{S}(\{\omega_{n-r-1},\omega_{n-r}\})=\mathcal{S}(\mathcal{X}_{n-r-2})^{\circ}.$

This process can be continued by 
using the restriction of $\bar{\omega}_{n-r-2}$ to $\bar{\Sigma}_{n-r-1},$ which also defines  a completely integrable  Pfaffian equation in a  $(n-2)$-dimensional space. The corresponding integral manifold is of dimension $n-3$  and it can imbedded into ${\Sigma}_{n-r-1}$ (and hence into ${\Sigma}_{n-r}$)  through the composition of the
local parametrizations of $\bar{\Sigma}_{n-r-1}$ and ${\Sigma}_{n-r}.$  By construction, this $(n-3)$-dimensional submanifold is  an integral manifold of  $\mathcal{S}(\Lambda_{n-r-3})=\mathcal{S}(\mathcal{X}_{n-r-3})^{\circ}.$
 In the last step we obtain a $r$-dimensional integral manifold 
of $\mathcal{S}(\Lambda_0)=\mathcal{S}(\mathcal{Z})^{\circ}.$
\end{proof}

Observe that the previous proof provides a systematic procedure to construct the integral manifolds. We refer to the reader to Example  \eqref{exampleorden3} (for instance) to see how this constructive procedure works in practice. 

Theorem \ref{maintheorem} extends the class of objects that can be used to integrate involutive distributions because $\mathcal{C}^{\infty}$-structures satisfy weaker conditions than solvable structures (see Remark \ref{weakercondition} (c)). The cost is that the corresponding Pfaffian equations given in Theorem \ref{maintheorem}, although  completely integrable, might not be defined by closed 1-forms, as in the case of solvable structures.

\section{$\mathcal{C}^{\infty}$-structures and differential equations }\label{applications}
\subsection{$\mathcal{C}^{\infty}$-structures for ODEs: General case}\label{scalarODE}

In this section we apply the results obtained in Section \ref{mainresults} to the particular case in which the involutive distribution $\mathcal{Z}$ corresponds to the distribution $\mathcal{S}(\{A\})$  generated by the vector field \eqref{campoA} associated to an $m$th-order ordinary differential equation of  the form \eqref{ODE1}.

A $\mathcal{C}^{\infty}$-structure $\langle X_1,\ldots, X_{m} \rangle$ for $\mathcal{S}(\{A\})$ will be simply called a $\mathcal{C}^{\infty}$-structure  of  equation \eqref{ODE1}. According to Definition \ref{Qsolvable}, the vector field  $X_1$ must be a $\mathcal{C}^{\infty}$-symmetry of the distribution $\mathcal{S}(\{A\}),$ i.e., a generalized $\mathcal{C}^{\infty}$-symmetry of (\ref{ODE1}) in the sense of Definition \ref{general_cinf-sym} (see Remark \ref{remark_cinfsym} (c)):
\begin{equation}\label{paraX1}
    [X_1,A]=\lambda^1 X_1 + c^1 A,
\end{equation}
for some functions $\lambda^1,c^1\in \mathcal{C}^{\infty}(M^{(m-1)}).$
The second vector field $X_2$ in the \cinf-structure must satisfy relationships of the form:
\begin{equation}\label{paraX2}
    \begin{array}{rcl}
   \left[X_2,A\right]  & =&  \lambda^2 X_2+d_1^2 X_1+c^2 A,\\
   \left[X_2,X_1\right]  & =&  \gamma_1^2 X_2+\widetilde{d}_{1}^2 X_1+\widetilde{c}_1^2 A.
\end{array}
\end{equation}

Similarly, for the remaining elements in the \cinf-structure, $X_i,$ for $i=1,\ldots,m,$ we must have
\begin{equation}\label{paraXj}
\begin{array}{rcl}
[X_i,A]&=&\lambda^i X_i+d_1^i X_1+\ldots+d_{i-1}^iX_{i-1}+c^iA,\\

[X_i,X_j]&=&\gamma_j^i X_i+\widetilde{d}_1^i X_1+\ldots+\widetilde{d}_{i-1}^i X_{i-1}+\widetilde{c}_j^iA \quad \mbox{for} \quad j<i,\end{array}
\end{equation}
where the coefficients are smooth functions on $M^{(m-1)}.$
 
As a consequence of  Remark \ref{remark_cinfsym} (d),  the last vector field $X_m$ can be any vector field pointwise linearly independent of $\{A,X_1,\ldots, X_{m-1}\}.$

The immediate application of Theorem \ref{maintheorem} yields the following theorem.

\begin{theorem}\label{estructuraODE}
A $\mathcal{C}^{\infty}$-structure for a $m$th-order equation  \eqref{ODE1} permits to  integrate the equation by solving successively $m$ completely integrable Pfaffian equations.
\end{theorem}

To apply this theorem we can proceed as follows. Let $\boldsymbol{\Omega}=dx\wedge du\wedge \ldots \wedge du_{m-1}$ be the volume form in $M^{(m-1)}.$ Given the $\mathcal{C}^{\infty}$-structure $\langle X_1,\ldots, X_{m} \rangle$ for $\mathcal{S}(\{A\}),$ we construct the following 1-forms: 
\begin{equation}\label{sigmas}
\begin{array}{lll}    
\omega_i&=&{X_1\,\lrcorner\,X_2\,\lrcorner\,\ldots\,\lrcorner\,\widehat{X_i}\,\lrcorner\,\ldots\,\lrcorner\,X_m\,\lrcorner\,A\,\lrcorner\,\boldsymbol{\Omega}}, \quad 1\leq i\leq m,
\end{array}
\end{equation}
where $\widehat{X_i}$ indicates omission of $X_i.$ It can be checked that $\langle \omega_1,\ldots,\omega_{m}\rangle$ is a \cinf-structure of 1-forms for  $\mathcal{S}(\{A\})$ in the sense of Definition \ref{Qsolvabledual},  and that, for $1\leq k\leq m,$ the condition  $\mathcal{S}(\mathcal{X}_k)^{\circ}=\mathcal{S}(\Lambda_k)$ is satisfied. 
As a consequence, the ideal $\mathcal{I}(\{\omega_m\})$ is a differential ideal, i.e., 
$$ 
\omega_{m} \wedge d\omega_{m}=0.
$$ 
Therefore, there exists a first integral  $I_m$ of $\mathcal{S}(\Lambda_{m-1}),$ that according to \eqref{edpintegralprimera} satisfies
\begin{equation}
    dI_m\wedge \omega_{m}=0.
\end{equation}
 Since $A\,\lrcorner\,\omega_m=0,$ the function $I_m$ is a first integral of $A.$ Also  note that $X_{k}(I_m)=0$ for $1\leq k\leq m-1.$ 
The corresponding level sets  of $I_m$ define  integral manifolds $\Sigma_m$ of the involutive distribution $\mathcal{S}(\mathcal{X}_{m-1})=\mathcal{S}(\{A,X_1,\ldots,X_{m-1}\}).$

Since $\mathcal{I}(\{\omega_{m-1},\omega_m\})$ is a differential ideal, the restriction of the 1-form $\omega_{m-1}$ to $\Sigma_m,$ which will be denoted by $\bar{\omega}_{m-1},$ is Frobenius integrable. As before,  a  first integral $\bar{I}_{m-1}$ of $\bar{\omega}_{m-1},$ arises from the condition 
\begin{equation}
    d\bar{I}_{m-1}\wedge \bar{\omega}_{m-1}=0.
\end{equation}
 The procedure continues in a similar way,  by integrating the  remaining completely integrable Pfaffian equations defined by the successive restrictions of the 1-forms in the \cinf-structure.  Observe that at each stage the  completely integrable Pfaffian equation has  one less variable than in the previous step.

In particular, the integrability by quadratures in terms of solvable structures (see Section \ref{symmetries}) follows from Theorem \ref{estructuraODE} for the special case in which all the vector fields $X_i$ in the $\mathcal{C}^{\infty}$-structure are symmetries of $\mathcal{S}(\mathcal{X}_{i-1}),$ for $1\leq i\leq m.$

\subsection{\cinf-structures for second-order ODEs}
For a second-order ODE
\begin{equation}\label{orden2}
u_2=\phi(x,u,u_1),
\end{equation} 
the first element $X_1$ of a \cinf-structure can be given by the first-order $\lambda$-prolongation of a known $\mathcal{C}^{\infty}$-symmetry of the equation. By considering its canonical representative   $(\partial_u,\lambda),$ it can be assumed without loss of generality that $X_1=(\partial_u)^{[\lambda,(1)]}=\partial_u+\lambda\partial{u_1}$ and hence $[X_1,A]=\lambda X_1$ (see \cite{muriel2009} for details). According to Theorem \ref{estructuraODE}, the corresponding 1-form $\omega_2$ defined in \eqref{sigmas},  \begin{equation}\label{sigmalambda1}
    \omega_2=X_1\,\lrcorner\,A\,\lrcorner\,\boldsymbol{\Omega}=(u_1\lambda-\phi)dx-\lambda du+du_1,
\end{equation} 
is Frobenius integrable. It can be checked that the condition $\omega_2\wedge d\omega_2=0$ implies that function $\lambda=\lambda(x,u,u_1)$ must be a solution of
\begin{equation}
    \label{ecudet}
    \lambda_x+\lambda_u u_1+\lambda_{u_1}\phi+\lambda^2=\phi_u+\lambda\phi_{u_1}
\end{equation} and conversely. Condition \eqref{ecudet} is precisely the determining equation that must satisfy $\lambda$ in order to $(\partial_u,\lambda)$
be a $\mathcal{C}^{\infty}$-symmetry of equation \eqref{orden2} \cite{muriel2009}.

By Remark \ref{involutive} (d) the second element of a  $\mathcal{C}^{\infty}$-structure can be any vector field $X_2$ such that $\{A,X_1,X_2\}$ is pointwise linearly independent. For instance, $X_2=\partial_{u_1}.$ In this case, $X_2\,\lrcorner\,X_1\,\lrcorner\,A\,\lrcorner\,\boldsymbol{\Omega}=\boldsymbol{\Omega}(A,X_1,X_2)=1.$ For the $\mathcal{C}^{\infty}$-structure $\langle X_1,X_2\rangle$ the corresponding 1-form \eqref{sigmas} becomes
\begin{equation}\label{sigmalambda2}
 \begin{array}{ll}  
 \omega_1&={X_2\,\lrcorner\,A\,\lrcorner\,\boldsymbol{\Omega}}=u_1dx-du.
 \end{array}\end{equation} 
In consequence, $\langle \omega_1,\omega_2\rangle$ is a $\mathcal{C}^{\infty}$-structure of 1-forms for equation \eqref{orden2}, where the 1-forms $\omega_1$ and $\omega_2$ are given in   \eqref{sigmalambda2} and \eqref{sigmalambda1}, respectively.

In examples \ref{example1lie} and \ref{exampleord2sinlie} we use this $\mathcal{C}^{\infty}$-structure to integrate second-order ODEs by using Theorem \ref{estructuraODE}.

The equation in Example \ref{example1lie} admits just one Lie point symmetry, that can be used to reduce the second-order equation to an Abel-type equation. For most of this kind of problems, the classical  Lie  method fails, due to the difficulty in solving the reduced equation or in evaluating the necessary quadrature  to recover the solution to the initial equation.

\begin{example}\label{example1lie}
Consider the second-order ODE
\begin{equation}\label{orden2_con1Lie}
u_2=-\dfrac{(u_1+1)(x+u) }{x u },\quad xu\neq0,
\end{equation}
which admits only the Lie point symmetry \begin{equation}\label{simetriadelie}
    \mathbf{v}=x \partial_x+u\partial_u.
\end{equation}
The classical order-reduction method leads to the Abel-type equation 
\begin{equation}\label{abel}
  w_y ={\frac { \left( y+1\right)^2  }{y}} w^{3}+{\frac {
 \left( 2\,y+1 \right)  }{y}w^{2}},
\end{equation} where \begin{equation}\label{cambioabel}
    y={\dfrac {u  }{x}}\quad\mbox{and}\quad w={\dfrac {x}{ u_1 x-u }}.
\end{equation}
 Maple system provides a rather involved implicit solution of \eqref{abel} in terms of the exponential integral function, from which it does not seem  possible to recover the general solution for \eqref{orden2_con1Lie}.

In order to find a $\mathcal{C}^{\infty}$-structure for equation \eqref{orden2_con1Lie} we first seek for a $\mathcal{C}^{\infty}$-symmetry  of the form $(\partial_u,\lambda).$ The solution  $\lambda=\dfrac{u_1+1}{u}$ for the corresponding determining equation \eqref{ecudet} can be easily determined. Therefore we can choose as the first element of the $\mathcal{C}^{\infty}$-structure the vector field   $$X_1=(\partial_u)^{[\lambda,(1)]}=\partial_u+\dfrac{u_1+1}{u}\partial_{u_1},$$ which satisfies $[X_1,A]=\lambda X_1.$ The corresponding form \eqref{sigmalambda1} becomes 
$$\omega_2=X_1\,\lrcorner\,A\,\lrcorner\,\boldsymbol{\Omega}=\dfrac{(u_1+1)(x+u+x u_1)}{x u}  dx-\dfrac{u_1+1}{u}du+d u_1.$$
According to Theorem \ref{estructuraODE} the Pfaffian equation $\omega_2\equiv 0$ is completely integrable. 

In order to integrate  $\omega_2\equiv 0$, we first observe that, for the admitted Lie point symmetry \eqref{simetriadelie}, we have that $V=\mathbf{v}^{(1)}=x\partial_x+u\partial_u$ satisfies
$$[V,X_1]=-X_1,$$ i.e., $V$ is a symmetry of the distribution $\mathcal{S}(\{A,X_1\}).$ In this situation, according to Theorem 2.1 in \cite{sherring1992geometric}, it is possible to use the symmetry $V$ to construct directly an   integrating factor of $\omega_2.$ Such integrating factor is given by \begin{equation}\label{mu2ex1}
    \mu_2(x,u,u_1)=\dfrac{1}{V\,\lrcorner\,\omega_2}=\dfrac{u}{x(1+u_1)^2}
\end{equation}  and a corresponding primitive for the exact 1-form $\mu_2\omega_2$ arises by a simple quadrature:
$$
I_2(x,u,u_1)=\ln(x)-\dfrac{u}{x(1+u_1)}.
$$
We can choose $X_2=\partial_{u_1}$ as the second element of a $\mathcal{C}^{\infty}$-structure  for the equation, because $\{A, X_1,X_2\}$ are pointwise linearly independent (see Remark \ref{involutive} (d)). Indeed, it can be checked that the following commutation relations hold: 
$$[X_1,A]=\lambda X_1,
\qquad [X_2,A]=X_1-\left(\dfrac{u_1x+2x+u}{xu}\right)X_2,
\qquad [X_2,X_1]=\dfrac{1}{u}X_2.$$ Observe that neither $X_1$ nor $X_2$ are symmetries of $\mathcal{S}(\{A\})$ and  that $X_2$ is not a symmetry of $\mathcal{S}(\{A,X_1\}).$

The second completely integrable Pfaffian equation associated to the $\mathcal{C}^{\infty}$-structure $\langle X_1,X_2\rangle$ is defined by the restriction of $\omega_1=X_2\,\lrcorner\,A\,\lrcorner\,\boldsymbol{\Omega}= u_1 d x-du$ to the hypersurface defined by $ I_2(x,u,u_1)=C_2,$ $C_2\in \mathbb{R}.$
Such hypersurface can be parametrized through 
$$
\iota_2: (x,u) \mapsto \left(x,u,\dfrac{  u}{x(
\ln(x)-C_2) }-1\right)
$$
and therefore
$$
\bar{\omega}_1:=\iota_2^*(\omega_1)=\left(\dfrac{u}{x(
\ln(x)-C_2) }-1\right) dx-du.
$$ 
The 1-form $\bar{\omega}_1$ is completely integrable, although is not closed. However, it is easy to find an integrating factor that only depends on $x:$ 
\begin{equation}\label{mu1ex1}
    \mu_1(x;C_2)= \frac{1}{\ln(x)-C_2}.
\end{equation} A corresponding primitive arises by a simple quadrature:
$$
I_1(x,u;C_2)=\frac{u}{C_2-\ln(x)}+\phi(x;C_2),$$
where $$
\phi'(x;C_2)=\frac{1}{C_2-\ln(x)},$$ and the prime denotes derivation with respect to $x.$

Consequently, if we solve for $ u $ in  $I_1(x,u;C_2)=C_1$, where $ C_1\in \mathbb{R} $, we obtain the explicit solution of (\ref{orden2_con1Lie}) in the form:
$$u(x)=(\ln(x)-C_2)(\phi(x;C_2)-C_1),\qquad C_1,C_2\in \mathbb{R}.$$
It is worth noting that the general solution of the Abel equation \eqref{abel} can be obtained as by-product of our procedure. By using \eqref{cambioabel}, the  general solution of the Abel equation \eqref{abel} can be expressed in parametric form as follows: 
    $$\begin{array}{lll}
         y(x) &=&  {\dfrac {(\ln(x)-C_2)(\phi(x;C_2)-C_1)  }{x}},\\
        w(x) &=&  -{\dfrac {x}{ \left(\ln  (x)-C_2 \right) \left(\phi(x;C_2)-C_1\right)- \phi(x;C_2) + C_1+x}}.
    \end{array}$$ 
\end{example}

In the following example we apply the \cinf-structure method to completely solve a second-order ODE lacking  Lie point symmetries.

\begin{example}\label{exampleord2sinlie}

It can be checked that the second-order ODE \begin{equation}\label{orden2sinLie}
u_2+\dfrac{x(u_1+x)}{u}+\dfrac{u}{x^4}+1=0, \quad xu\neq 0
    \end{equation}
does not admit Lie point symmetries. However, it is not difficult to find a particular solution of the form $\lambda=a(x,u)u_1+b(x,u)$ to the corresponding determining equation \eqref{ecudet}, giving rise to the $\mathcal{C}^{\infty}$-symmetry $(\partial_u,\lambda)$ defined by the function
$\lambda=\dfrac{u_1+x}{u}.$
Therefore
$$
X_1=(\partial_u)^{[\lambda,(1)]}=\partial_u+\dfrac{u_1+x}{u}\partial_{u_1}
$$
can be chosen as the first vector field of a \cinf-structure for \eqref{orden2sinLie}. 

The corresponding 1-form
$\omega_2=X_1\,\lrcorner\,A\,\lrcorner\,\boldsymbol{\Omega}$ becomes
\begin{equation}\label{sigma2ej2}
    \omega_2=\left(\dfrac{(u_1+x)^2}{u}+\dfrac{u}{x^4}+1\right)dx-\dfrac{u_1+x}{u}du+du_1.
    \end{equation}

The condition $dI_2\wedge \omega_2=0$ provides a determining system of partial differential equations for a first integral $I_2$ of $\mathcal{S}(\{\omega_2\}),$ where the subscripts denote the corresponding partial derivatives: 
\begin{align}
\label{una}  (I_2)_u+ \left( \dfrac{x+u_1}{u}\right) \, (I_2)_{u_1} & =0,\\
\label{dos}    (I_2)_x+  \left(\dfrac{(u_1+x)^2}{u}+\dfrac{u}{x^4}+1\right) \, (I_2)_{u_1}  & =0.
\end{align}
The general solution to equation \eqref{una} becomes \begin{equation}\label{ansatz}
    I_2(x,u,u_1)=F\left(x,\dfrac{x+u_1}{u}\right),
\end{equation} where $F=F(x,t)$ is an arbitrary smooth function. From equation \eqref{dos} it follows that the function $F=F(x,t)$ must satisfy: 
\begin{equation}
\left(\dfrac{1}{x^4}+t^2\right) F_t -F_x=0,
\end{equation}
which admits the particular solution $F(x,t)=\arctan\left({t x^2-x}\right)-\dfrac{1}{x}.$ Therefore, by \eqref{ansatz}, the function 
\begin{equation}
  I_2(x,u,u_1)=\arctan\left(\dfrac{x(xu_1+x^2-u)}{u}\right)-\dfrac{1}{x} 
\end{equation} is a first integral of $\mathcal{S}(\{\omega_2\}).$ Therefore  $\Sigma_2:=\{(x,u,u_1)\in M^{(1)}: I_2(x,u,u_1)=C_2\},$ where $C_2 \in \mathbb{R},$ is an integral manifold of $\mathcal{S}(\{A, X_1\}).$

In order to complete the $\mathcal{C}^{\infty}$-structure, let us recall that by Remark \ref{involutive} (d), $X_2$ may be any vector field such that $\{A, X_1,X_2\}$  is a set of pointwise linearly independent vector fields. If we choose the vector field $X_2=\partial_{u_1}$  then 
$\omega_1=u_1dx- du.$ 

We use the parametrization 
$$\iota_2(x,u)=\left(x,u, \dfrac{u}{x^2}{\tan\left(C_2+\dfrac{1}{x}\right)}-x+\dfrac{u}{x}\right)$$ of the integral manifold $\Sigma_2$ to calculate 
\begin{equation}\label{sigma1res}
 \bar{\omega}_1:=   \iota_2^*(\omega_1)=\left( \dfrac{u}{x^2}{\tan\left(C_2+\dfrac{1}{x}\right)}-x+\dfrac{u}{x}\right)dx-du.
\end{equation}
Although \eqref{sigma1res} is not closed, it is easy to calculate an integrating factor that only depends on $x$:
\begin{equation}
\mu_1(x;C_2)= \dfrac{1}{x\cos\left(C_2+\dfrac{1}{x}\right)}.    
\end{equation}
A primitive for the exact form $\mu_1\bar{\omega}_1$ becomes 
\begin{equation}
    I_1(x,u;C_2)=-\dfrac{u}{x \cos\left(C_2+\dfrac{1}{x}\right)}+\psi(x;C_2),
\end{equation} where $\psi(x;C_2)$ satisfies 
$$\psi'(x;C_2)=-\dfrac{1}{\cos\left(C_2+\dfrac{1}{x}\right)}.$$ 
From the implicit solution  $I_1(x,u;C_2)=C_1,$ $C_1\in \mathbb{R},$ we finally obtain the explicit solution for equation \eqref{orden2sinLie} in the form
\begin{equation}
    u(x)=x \cos\left(C_2+\dfrac{1}{x}\right)\left(\psi(x;C_2)-C_1\right),\quad C_1,C_2\in \mathbb{R}.
\end{equation}

\end{example}

\subsection{\cinf-structures for third-order equations: an example}

In the following example we consider a third-order equation which admits just one Lie point symmetry. The associated second-order reduced equation does not admit Lie point symmetries. We first construct a \cinf-structure of the third-order ODE and show how it can be used to integrate the equation completely.  

\begin{example}\label{exampleorden3}

The third-order ODE 
\begin{equation}\label{orden3}
u_3=\frac{x u_2 e^{x}+2 u_1 e^{x}-3 u_2 e^{x}+3 u_2^{2}}{x e^{x} +2 u_1}
\end{equation} 
only admits the Lie point symmetry $\partial_u$, which provides, by means of the transformation $y=x$ and $w=u_1,$ the  obvious order reduction
\begin{equation}\label{reduorden2}
  w_2=\dfrac{y w_1 e^y+2 w e^y-3 w_1 e^y+3 w_1^2}{y e^y+2w}.  
\end{equation}
 It can be checked that equation \eqref{reduorden2} does not admit Lie point symmetries, and consequently, the classical Lie method cannot be applied to solve equation \eqref{orden3}.

In order to construct a $\mathcal{C}^{\infty}$-structure for equation (\ref{orden3}) we can  use the Lie point symmetry $\partial_u$ to get the first element $X_1:=(\partial_u)^{(2)}=\partial_u,$ for which $[X_1,A]=0$ (i.e., the functions $\lambda^1$ and $c^1$ in \eqref{paraX1} are identically null).

For a second element of the $\mathcal{C}^{\infty}$-structure, we can consider a vector field of the form $$X_2 = \partial_{u_1} + \eta(x,u_1,u_2) \partial_{u_2},$$ 
which satisfies $[X_1,X_2]=0$ (i.e., the functions $\gamma_1^2,\bar{d}_1^2$ and  $\bar{d}_1^2$ in \eqref{paraX2} are identically null). The determining equation for the function $\eta$ results from the condition that the distribution $\mathcal{S}(\left\{ A, X_1, X_2\right\})$ must be involutive. By Proposition \ref{caracterizacioninvolutivo} this condition can be formulated in terms of corresponding  1-form $\omega_3$ defined in \eqref{sigmas} as $\omega_3\wedge d\omega_3=0.$  For simplicity we search for a particular solution of the determining equation of the form $\eta(x,u_1,u_2) = a(x,u_1)u_2+b(x,u_1),$ giving rise to 
$$\eta(x,u_1,u_2) = \displaystyle\frac{2(u_2-e^x)}{xe^x+2u_1}.$$
The last element $X_3$ of a $\mathcal{C}^{\infty}$-structure can be any vector field such that the vector fields  $A,X_1,X_2$ and $X_3$  are pointwise linearly independent (see Remark \ref{involutive} (d)). Therefore, the ordered set $\langle X_1,X_2,X_3\rangle$ defined by 
$$
X_1=\partial_u,\quad
X_2=\partial_{u_1}+\dfrac{2(u_2-e^x)}{x e^x+2 u_1} \partial_{u_2},\quad
X_3=\partial_{u_2}
$$
is  a $\mathcal{C}^{\infty}$-structure for equation (\ref{orden3}). This  can be checked by simply computing the corresponding Lie brackets.

Once a  $\mathcal{C}^{\infty}$-structure has been determined, we can use the algorithm described in Section \ref{applications} to integrate the given equation. Firstly, we  calculate the corresponding 1-forms (\ref{sigmas}), which in this example become:
$$
\begin{array}{l}
\omega_1=-u_1 dx+du,\quad 
\omega_2=u_2 dx - du_1, \\ \\
\omega_3=-\displaystyle\frac{(x u_2 e^x+2u_1 e^x-u_2 e^x+u_2^2)}{xe^x+2u_1}dx+\displaystyle\frac{2(e^x-u_2)}{xe^x+2u_1}du_1+ du_2.
\end{array}
$$
According to Theorem \ref{maintheorem} the Pfaffian equation $\omega_3\equiv 0$ is completely integrable. The condition $dI_3\wedge \omega_3=0$ provides the following system of determining equations for a first integral $I_3$ of $\mathcal{S}(\{\omega_3\}):$ 
\begin{align}
    \label{ec1} (I_3)_u&=0  ,\\
    \label{ec2} (xe^x+2u_1)(I_3)_{u_1}+2(u_2-e^x)(I_3)_{u_2}&=0 ,\\
    \label{ec3} (xe^x+2u_1)(I_3)_{x}+(u_2^2+xe^xu_2+2u_1e^x-e^xu_2)(I_3)_{u_2}&=0.
    \end{align}
It follows from equations \eqref{ec1} and \eqref{ec2} that
\begin{equation}\label{formaI3}
    I_3=F\left(x,\dfrac{u_2-e^x}{e^x+2u_1}\right),
\end{equation} where $F=F(x,t)$ is, in principle, an arbitrary smooth function. By equation \eqref{ec3} the function $F=F(x,t)$ must satisfy 
$$ t^2F_t+F_x=0,$$ whose general solution becomes $F(x,t)=G\left(x+\dfrac{1}{t}\right).$ In consequence, by \eqref{formaI3}, a particular solution to \eqref{ec1}-\eqref{ec3} is given by:
\begin{equation}
    I_3(x,u,u_1,u_2) = \displaystyle\frac{xu_2 + 2u_1}{u_2-e^x}.
\end{equation}

The next step in our procedure is to restrict $\omega_1$ and $\omega_2$ to the manifold implicitly defined by $I_3(x,u,u_1,u_2)=C_3,$ by defining the imbbeding
$$
\iota_3:(x,u,u_1) \mapsto \left(x,u,u_1, \frac{C_3 e^x+2 u_1}{C_3-x}\right).
$$
We obtain 
$$\begin{array}{l}
\bar{\omega}_1=\iota_3^*(\omega_1)=-u_1 dx+du,\quad
\bar{\omega}_2=\iota_3^*(\omega_2)=\displaystyle\frac{C_3e^x+2u_1}{C_3-x} dx - du_1.\\
\end{array}
$$
According to Theorem \ref{maintheorem} the Pfaffian equation $\bar{\omega}_2\equiv 0$ is completely integrable, i.e.,
$ \bar{\omega}_2 \wedge d \bar{\omega}_2 =0.$ In this case an integrating factor only depending on $x$ can be easily found: 
$$\mu_2(x;C_3)= -(C_3-x)^2.$$
A corresponding primitive of $\mu_2\bar{\omega}_2$ can be calculated by quadrature and it is given by
$$I_2(x,u,u_1;C_3) = C_3e^x(x-C_3-1)+u_1 (x-C_3)^2.$$
In order to restrict  $\bar{\omega}_1$  to $I_2(x,u,u_1;C_3)=C_2,$ we solve for $u_1$ in this last equation and define
$$
\iota_2:(x,u)\mapsto \left(x,u,\dfrac{C_3e^x(C_3-x+1)+C_2}{(C_3-x)^2}\right).
$$
The pullback  of $\bar{\omega}_1$ by $\iota_2$ becomes  
$$
\widetilde{\omega}_1:=\iota_2^*(\bar{\omega}_1)= -\dfrac{C_3e^x(C_3-x+1)+C_2}{(C_3-x)^2}dx+du. 
$$
In this case the 1-form $\widetilde{\omega}_1$ is not only Frobenius integrable, but closed, i.e.,  
$d \widetilde{\omega}_1 =0.$ 
Therefore the 1-form  $\widetilde{\omega}_1$ is  locally exact and  a corresponding primitive can be computed by simple quadrature:
$$I_1(x,u;C_3,C_2)=\displaystyle\frac{C_3(e^x-u)+ux+C_2}{x-C_3}.$$
Finally, from the implicit solution $I_1(x,u;C_3,C_2)=C_1$ we can obtain the exact general solution of equation \eqref{orden3} in explicit form: 
$$
u(x)=\dfrac{C_1(C_3-x)+C_3e^x+C_2}{C_3-x}, \qquad C_1,C_2, C_3 \in \mathbb{R}.
$$

It is interesting to  observe that, as by-product, our procedure also provides the general solution of the second-order equation \eqref{reduorden2}: $$w(y)=\dfrac{C_2-C_3e^y(y-C_3-1)}{(y-C_3)^2}.$$ Recall that equation \eqref{reduorden2} does not admit Lie point symmetries and, as far as a we know, cannot be solved by standard procedures. 
\end{example}

\subsection{\cinf-structures for systems of ODEs}

In this section we show how the results that have been presented in Section \ref{scalarODE}  for scalar ODEs can be  adapted and applied to systems of ODEs.

It is well-known that by introducing auxiliary dependent variables, if that is necessary, any system of ODEs can be transformed into a system of first-order ODEs. We also assume that the system can  be  written in  the normal form: 
\begin{equation}
    \label{sistema1}
u^i_1=F^j(x,u^1,\ldots,u^r), \quad 1\leq i\leq r.
\end{equation} We consider the vector field associated to system \eqref{sistema1}: 
\begin{equation} \label{sistemaA1}
    A=\partial_x+\displaystyle\sum_{i=1}^r F^i(x,u^1,\ldots,u^r){\partial}_{u^i}.
\end{equation}
A \cinf-structure for the distribution generated by $A,$ in the sense of Definition \ref{solvable}, will be called a \cinf-structure for  \eqref{sistema1}. In this case, a  \cinf-structure is an ordered collection of $r$ smooth vector fields $\langle X_1,\ldots,X_r\rangle$ on some open set $M$ of $\mathbb{R}^{r+1}$ that must satisfy relationships of the form \eqref{paraX1}-\eqref{paraXj}, where now $A$ is given by \eqref{sistemaA1}.  The last element $X_r,$ as in the case of scalar ODEs, may be any vector field such that $\{A, X_1,\dots,X_r \}$ is a set of pointwise linearly independent vector fields on $M.$
A $\mathcal{S}(\{A\})$-related \cinf-structure of 1-forms for  system \eqref{sistema1} can be explicitly constructed by using the volume form $\boldsymbol{\Omega}=dx\wedge du^1\wedge \ldots\wedge du^r:$  
\begin{equation}\label{sigmassis}%
\begin{array}{lll}    
\omega_i&=&{X_1\,\lrcorner\,X_2\,\lrcorner\,\ldots\,\lrcorner\,\widehat{X_i}\,\lrcorner\,\ldots\,\lrcorner\,X_r\,\lrcorner\,A\,\lrcorner\,\boldsymbol{\Omega}}, \quad 1\leq i\leq r.
\end{array}
\end{equation}
Such  \cinf-structure of 1-forms permits to solve the system through the integration of $r$ successive completely integrable Pfaffian equations, as stated in Theorem \ref{maintheorem}. This procedure is illustrated in the following example.

\begin{example}\label{examplesystem}

Consider the following system
\begin{equation}\label{eq_ej_sistem}
\left\{\begin{array}{l}
\dfrac{d x(t)}{dt}=\dfrac{t y+x-2 y\tan(t)}{t-\tan(t)}, \\

\dfrac{d y(t)}{dt}=\dfrac{x+(t x-t^2 y-y) \tan(t)}{t^2 -t\tan(t)}
\end{array}\right.
\end{equation}
and denote by $A$ its associated vector field \eqref{sistemaA1}. 

The first  element  of   a \cinf-structure for  $\mathcal{S}(\{A\})$ can be searched, for simplicity,  as a vector field of the form
$X_1=\partial_x+\eta(t) \partial_y.$ The condition $[X_1,A]\in \mathcal{S}(\mathcal\{A,X_1\})$  can be expressed in terms of the corresponding 1-form  $\omega_2=X_1\,\lrcorner\,A\,\lrcorner\,\boldsymbol{\Omega}$ given in \eqref{sigmassis} as $\omega_2\wedge d\omega_2=0.$ This provides the following determining equation for $\eta(t):$
$$
\eta'(t)=\frac{2  \tan (t) -t}{t-\tan (t)}\eta(t)^{2}-\frac{ t^{2} \tan (t)+t +\tan (t) }{t^{2}-t \tan (t)}\eta(t)+\frac{t\tan (t) +1}{t^{2}-t\tan (t) },
$$
which  admits the particular solution $\eta(t)=1/t$.

The second element of the \cinf-structure can be any vector field pointwise linearly independent with $\{A,X_1\}$ (Remark \ref{involutive} (d)). We can choose,  for instance, $X_2=\partial_y.$  Therefore, the ordered set $\langle X_1,X_2  \rangle$ defined by the vector fields 
$$
X_1={\partial_x}+{\dfrac{1}{t}\partial_y}, \quad
X_2=\partial_y
$$
constitutes a \cinf-structure for $\mathcal{S}(\{A\})$.

A  $\mathcal{S}(\{A\})$-related \cinf-structure of 1-forms (i.e., satisfying conditions \eqref{eq_condiciondual}), 
is given by the differential forms corresponding to \eqref{sigmassis}
\begin{equation}\label{omegassistema}
    \begin{array}{l}
    \omega_1=\dfrac{yt+x-2y \tan (t)}{t-\tan (t)} d t-d x,   \\
    \omega_2=\dfrac{y t^{2}\tan (t) +y t-x t \tan (t) - y\tan (t) }{t^2 -t\tan (t)} d t-\dfrac{1}{t} d x+d y.
    \end{array}
\end{equation}

By Theorem \ref{mainresults}, the Pfaffian equation $\omega_2\equiv 0$ is completely integrable. In order to find a first integral of $\mathcal{S}(\{\omega_2\})$ we use the condition $dI_2\wedge \omega_2=0$ and obtain:
\begin{align}
    \label{eq1I2sist} 
  (I_2)_y+t (I_2)_x=0 ,\\
      \label{eq2I2sist}   \left(y t^{2}\tan (t) +y t-x t \tan (t) - y\tan (t) \right)(I_2)_x + \left(t- \tan (t)\right)(I_2)_t=0.
   \end{align}

From \eqref{eq1I2sist} it follows that
$$
I_2(t,x,y)=F\left(t,\dfrac{yt-x}{t}\right),
$$
where $F(t,s)$ is, in principle, an arbitrary smooth function. By using \eqref{eq2I2sist} we conclude that $F$ must satisfy the partial differential equation

\begin{equation}\label{eq2I2sislineal}
    s\left( \tan (t)-\tan (t) t^{2}-t\right) F_s+ t( t- \tan (t) )F_t=0,
\end{equation}
whose  general solution becomes
$$
F(t,s)=G\left( \frac{t s}{t \cos(t)-\sin(t)}\right),
$$
where $G$ is an arbitrary smooth function. Therefore, a particular solution for \eqref{eq1I2sist}-\eqref{eq2I2sist} is 
\begin{equation}\label{I2sis}
    I_2(t,x,y)=\dfrac{yt -x}{t \cos(t)-\sin(t)}.
\end{equation} 

The level set  given by $I_2(t,x,y)=C_2,$ with $ C_2 \in \mathbb{R},$ defines an  integral manifold $\Sigma_2$ of the distribution $\mathcal{S}(\{A,X_1\}).$ A local parametrization is defined by the imbedding 
$$
\iota_2:(t,x) \mapsto \left(t,x,\dfrac{C_2 t \cos(t)-C_2\sin(t)-x}{t}\right).
$$
The restriction of the 1-forms in \eqref{omegassistema} to $\Sigma_2$  become $ \iota_2^*(\omega_2)=0$ (because $\Sigma_2$ is an integral manifold of $\mathcal{S}(\{\omega_2\})$) and 
$$
\bar{\omega}_1:=\iota_2^*(\omega_1)=\dfrac{C_2 t \cos(t)-2C_2\sin(t)+2x}{t} dt-dx.
$$

The Pfaffian equation $\bar{\omega}_1 \equiv 0 $ is also completely integrable. The condition $dI_1\wedge \bar{\omega}_1=0$ gives rise to the equation
\begin{equation}\label{eq2I1sislineal}
\left(2 C_2 \sin (t)-C_2 t\cos (t)-2 x\right) (I_1)_x-t(I_1)_t=0,
\end{equation}
which admits the particular solution
\begin{equation}\label{I1sis}
I_1(t,x,y;C_2)=\dfrac{x-C_2\sin(t)}{t^2}.
\end{equation}

This function is a first integral of  $\bar{A}= (\iota_1^{-1})_*\left({A|}_{{\Sigma}_{2}}\right).$ Therefore $I_1(t,x,y;C_2)=C_1$, with $C_1 \in \mathbb{R}$, defines an integral manifold $\bar{\Sigma}_1$ 
of $\mathcal{S}(\bar{A}),$ which can be parametrized through   the map
$$
\iota_1:t \mapsto (t,C_1 t^2+C_2 \sin(t)).
$$
Finally, $\Sigma:=\iota_2(\bar{\Sigma}_1)$ is an integral manifold of $\mathcal{S}(\{A\})$ that can be parametrized by $\iota:=\iota_2\circ\iota_1:$
$$
\iota: t \mapsto (t,C_1 t^2+C_2 \sin(t),C_2 \cos(t)+C_1 t).
$$
Therefore, $$x(t)=C_1 t^2+C_2 \sin(t),\,y(t)=C_2 \cos(t)+C_1 t,
$$ where  $C_1,C_2\in \mathbb{R},$ is the general solution of system \eqref{eq_ej_sistem}.

Observe that in this example, the first integral $I_2$ in \eqref{I2sis} associated to $\omega_2\equiv 0$ has been obtained by solving the PDEs \eqref{eq1I2sist} and \eqref{eq2I2sislineal}, while the first integral $I_1$ in \eqref{I1sis} associated to $\omega_2\equiv 0$ has been obtained from equation \eqref{eq2I1sislineal}. It can be checked that the characteristic systems associated to these three PDEs correspond to first-order ODEs that are linear, and they are therefore easy to integrate.  This shows that,  although the Pfaffian equations associated to a \cinf-structure  might not be defined by closed 1-forms (as in the case of solvable structures),  there are examples, as in this case, in which such Pfaffian equations  are simple and can be solved by standard procedures. 
  
\end{example}

\section{Discussion and related approaches}
	
	For any involutive distribution $\mathcal{Z}$ we have introduced a new structure of vector fields, called \cinf-structure,   which plays an outstanding role in the determination of the integral manifolds of $\mathcal{Z}$ (Theorem \ref{maintheorem}). For the particular case where  $\mathcal{Z}$ is defined by a single vector field associated with an ODEs system, the new approach greatly enlarges the types of vector fields that have been considered in the last  decades with the ultimate goal of integrating or reducing the underlying differential equations.
	
	In this section we present a short review and synthesis of  some relevant aspects of  approaches to that goal. First of all, it should be clear that these theories have a long history whose origins can be traced back to the  works of F. G. Frobenius and S. Lie. The main ideas have been  gradually evolving over the years by following different objectives and motivations. Some of the Lie's ideas derived to an algebraic synthesis that led to a new branch of Mathematics (the algebraic Lie theory) whose objectives are very different to the original ones. 
	
	The geometric theory of differential equations  emerged thanks to the works of S. Lie,  E. Cartan and others in the first half of the 20th century.  The methods in this theory are somewhat different from the classical (analytical) methods. This makes it difficult  
	an exhaustive comparison and  better understanding   of the achievements in both theories (analytical and geometrical).  This would be,  undoubtedly, an arduous task, because such approaches usually emerge from different motivations, pursue diverse objectives  and, hence, are presented and investigated from different perspectives, by using different mathematical languages and notations. Nevertheless,  we believe that these theories are mutually imbricated and,  despite their differences, any of them can offer a very positive influence on the others. In other words,   any effort to clarify, reinterpret, combine and complement the new results has certainly beneficial effects to the different approaches.

	In the context of the present paper, one of the most important concepts that extend the notion  of solvable algebra of local symmetries is the idea of solvable structure.  According to P. Basarab-Horwath \cite{basarab}, the origins of the concept of solvable
	structure (in the scheme of Frobenius approach) can be traced back to the
	works of Sophus Lie.  At the begining of the 90s,
	thanks to the works of P. Basarab-Horwath \cite{basarab}, J. Sherring and G. E. Prince \cite{sherring1992geometric}, T. Hartl and C. Athorne \cite{hartl1994solvable}, M.A. Barco and G.E. Prince \cite{Barco2001}, and others, the theory of solvable structures received new impulses and inspired many developments from both  analytical and geometrical  point of views.

	With the objective of extending Lie reduction methods to find exact solutions of equations not possessing Lie point symmetries, the concept of  $\lambda$-prolongation, and its characterization in terms of the  invariants-by-derivation property (IBDP) were introduced at the beginning  of the twenty-first century \cite{muriel01ima1, MR2001541}.  E. Pucci and G. Saccomandi provided in \cite{pucci}  an interesting geometric interpretation of these prolongations and introduced the so-called telescopic vector fields. In 2007 D. Catalano-Ferraioli 
proved that these prolongations can be interpreted as  shadows of some nonlocal symmetries, by using the technique of the $\lambda$-coverings \cite{catalano2007}.
	Along the first decade of this century,   G. Gaeta, P. Morando, G. Cicogna and S. Walcher explored several methods of prolongation of vector fields that generalized the concept of $\lambda$-prolongation. Collectively, these prolongations were called twisted prolongations (see the review papers \cite{gaetatwisted} and \cite{gaetacollective} by G. Gaeta and references therein). Depending on the context where the prolongations  are being considered, these generalizations and their corresponding notions of symmetry,  received specific names, according to   the auxiliary object they use. Thus, for instance,  the auxiliary object  for $\lambda$-symmetries of an ODE, is a scalar function $\lambda$ on the jet space; for $\mu$-symmetries of a PDE, the auxiliary object is a semi-basic one-form $\mu=\Lambda_idx^i$, where the $\Lambda_i$ are matrix functions, satisfying  the Maurer-Cartan equation \cite{gaetamorando,gaeta2,morando_mu_symmetries}. Similar extensions were carried out for dynamical systems, and specifically for Hamiltonian systems,  \cite{CicognaHamiltonian,cicogna2013dynamical} and for variational problems \cite{murielolver, cicognagaeta_noether, adrianJPA2018}. A review on these issues appears in \cite{muriel_evolution}.

	In 2012, G. Cicogna, G. Gaeta and S. Walcher introduced the concept of $\sigma$-symmetry in the study of the integrability of systems of ODEs \cite{cicogna2012generalization}.  These authors, in 2013, used $\sigma$-symmetries to separately study  dynamical systems \cite{cicogna2013dynamical} and orbital reducibility \cite{orbital}.  The initial concept is based on the notion of $\sigma$-prolongation (also called joint $\lambda$-prolongation) of a set of  $r$ vector fields defined on the 0th-order jet space $M$   \cite[Definition 2]{cicogna2012generalization}.  In this case, the auxiliary object is a $r\times r$ matrix whose coefficients are smooth functions $\sigma^i_{j}$ on $M^{(1)}.$
	An important aspect about  $\sigma$-prolonged vector fields is that, as in the case of $\lambda$-prolongations, they maintain the IBDP and therefore the $\sigma$-symmetries are as useful as standard ones for reduction purposes. This fact is the key to adapt the original Lie symmetry (or the $\lambda$-symmetry) reduction method, i.e.,   to reduce a system of ODEs which is invariant under a set of $\sigma$-prolonged vector fields. 	
	However, as it is remarked in 
	\cite[Section 3]{cicogna2012generalization}, the IBDP for $\sigma$-prolonged  vector fields, and hence the associated reduction procedure,  require that the  $\sigma$-prolonged  vector fields must constitute an involutive system. Assuming that the base vector fields  are in involution, the authors provide  a sufficient condition for the $\sigma$- prolongations be in involution, although the commutation relationships may be in general different from the original ones.

	   Among  the  questions brought forward in  	\cite[Section 6]{cicogna2012generalization}, the authors propose  the  
	 possible extensions  of the $\sigma$-symmetry theory to the study of Euler-Lagrange equations  and to the study of nonlocal symmetries,  in the line of the works by D. Catalano-Ferraioli and P. Morando  on solvable structures for ODEs \cite{catalano_nonlocal, Catalano_Paola}. These authors adapted solvable structures to local and nonlocal symmetries (classical and higher), introducing the notion of nonlocal solvable structure. For the problem of determining solvable structures of certain covering system, which in practice is usually a difficult one, the authors observed that previously known symmetry generators can help to find them.
  
  The authors of \cite{cicogna2012generalization} did also encourage the investigation of the use of $\sigma$-symmetries for a  possible generalization of solvable structures. As far as we know, this is  still  an open problem. Observe that, in some sense, the results in this paper  on \cinf-structures are aimed in that direction, although strictly speaking they are not based  on $\sigma$-symmetries, as suggested in \cite{cicogna2012generalization}, but on \cinf-symmetries of distributions.

	In \cite{paola.frobenius}, P. Morando presented a dimensional reduction method for a  vector field $Z$ defined on an $n$-dimensional manifold, based on the knowledge of a suitable integrable distribution $\mathcal{D}_0$ of $r<n$ vector fields.   The reconstruction  of the general solution to the system associated with $Z$ from the solutions of the reduced system was also analyzed. In particular, if $\mathcal{D}_0$ is a solvable structure for $Z$, then  the integral lines of $Z$ can be reconstructed by quadrature from the reduced ones. This kind of distributions are named in \cite{paola.frobenius} {\it{partial solvable structures}}. For the case of systems of ODEs, the reduction procedure based on $\sigma$-symmetries fits in with this general reduction method. In the particular case in which the $\sigma$-symmetry is a partial solvable structure (called a triangular $\sigma$-symmetry), the reconstruction problem can be solved by quadratures.   
 
In \cite{conmorando,mmr} the integrability of any $\mbox{SL}(2,\mathbb{R})$-invariant $n$th-order ODE  was tackled  by considering solvable structures with respect to the involutive distribution formed by the vector field associated to the equation and the symmetry generators of the nonsolvable symmetry algebra $\mathfrak{sl}(2,\mathbb{R})$. This procedure was later extended for any $n$th-order ODE with a known $k$-dimensional symmetry algebra $\mathcal{G}_k,$ solvable or not, with $1\leq k<n$ (\cite{MurielAIP}). These new structures were named {\it{ generalized solvable structures}} with respect to  $\mathcal{G}_k.$ 

As remarked by G. Gaeta in \cite{gaeta_frobenius}, the above described theories show that the focus must be on distributions rather than on a single vector field, and in involution rather than in the Lie algebraic structure. 


As a summary,   some of the most  significant differences  between \cinf-structures and the previous types of systems of vector fields are highlighted  below. 
	\begin{enumerate}
		\item  The concepts of solvable algebra, solvable structure and \cinf-structure  are established for an involutive distribution $\mathcal{Z}$ of $r$ vector fields on a $n$-dimensional manifold through an ordered set of $n-r$ vector fields $ \langle X_1,\ldots, X_{n-r}\rangle .$  For a solvable structure, as well as for a solvable algebra, the first element  $X_1$ must be a symmetry of $\mathcal{Z};$ this condition is not required for a  \cinf-structure;
	  in the commutation relationships between $X_1$ and the fields in $\mathcal{Z}$ the vector field $X_1$ may appear. Similar situations may occur for the successive elements.
	 \item For a solvable algebra, the Lie bracket of two generators must be a linear combination, with constant coefficients, of the generators.  In particular, the generators constitute an involutive system and have a complete system of joint invariants.  For solvable structures or \cinf-structures, that is not the case.
	  In the commutation relationships between an  $X_i$  and an $X_{i+j}$ there may appear the vector fields  in  $\mathcal{Z}$ and the  coefficients are, in general, non constant functions.
	
		\item The problem of the integration of the distribution $\mathcal{Z}$ by using successive quadratures of locally exact  one-forms  can be solved either with  solvable algebras or with solvable structures.
		The results in this paper show that with a \cinf-structure the integration problem can be splitted into successive integrable Pfaffian equations, although not necessarily integrable by quadratures.

    \item For both partial and generalized solvable structures, the involved vector fields are still symmetries of certain distributions containing the vector field associated to the equation. Nevertheless, it must be remarked that this condition is not required for a \cinf-structure.
    
\item As it has been said, for a \cinf-structure it is allowed that the vector fields of the  distribution  $\mathcal{Z}$ may appear in the commutation relationships. 
	In particular,  the vector fields in a \cinf-structure not necessarily must be in involution. Hence, the integration process carried out in this paper is not based on the IBDP (as for $\sigma$-symmetries); in fact it is more related to the methods based on  solvable structures, as presented in \cite{basarab,hartl1994solvable,Barco2001,Barco2002,BarcoPDE,Vecchi}. In consequence, \cinf-structures permit to construct the integral manifolds by solving successively $n-r$ completely integrable Pfaffian equations. At each step, the space  where each Pfaffian equation is defined has one less dimension than in the previous stage.

\end{enumerate}

\section{Final remarks}

Frobenius Theorem for involutive distributions of vector fields guarantees the existence of integral manifolds through each point, although it does not provide a constructive procedure to find them. Among the several concepts that have been introduced in the literature to obtain constructive integration methods, the concept of solvable structure is one of the  most powerful. Once a solvable structure has been determined, there is a systematic procedure to integrate the underlying distribution by successive quadratures. The main difficulty of this approach may be the determination of the elements of the solvable structure. These elements must be symmetries of the successive involutive distributions associated to the solvable structure.

In this work the concept of $\mathcal{C}^{\infty}$-symmetry, initially introduced for ODEs, has been adapted for involutive distributions of vector fields. This permits to extend the concept of solvable structure and introduce the notion of \cinf-structure for an involutive distribution.

These new objects play a fundamental role in the integrability of the distribution, with a somewhat similar functionality to  $\mathcal{C}^{\infty}$-symmetries in the integration of  ODEs. In this paper it has been proved that the knowledge of a $\mathcal{C}^{\infty}$-structure for a given rank $r$ involutive distribution allows to compute their integral manifolds by solving $n-r$ completely integrable Pfaffian equations. A solvable structure is  a particular case of a $\mathcal{C}^{\infty}$-structure.  Thus, $\mathcal{C}^{\infty}$-structures enlarge the classes of vector fields that can be used to integrate involutive distributions of vector fields. 

The main theoretical results have been applied to the particular case of  distributions generated by the vector fields associated to systems of ordinary differential equations. In particular,  it has been proved that a  $\mathcal{C}^{\infty}$-structure for a given $m$th-order scalar ODE leads to the integration of the equation after solving successively $m$ completely integrable Pfaffian equations. Each of these Pfaffian equations is defined in an space with one less dimension than in the previous stage.  The presented examples show that, although these  Pfaffian equations might not be defined by closed 1-forms, as in the case of solvable structures, they can be equally integrated. This has permitted to integrate completely the considered equations, even when the classical Lie method fails.

The possible limitations in the application of the presented approach come from the potential difficulty in finding particular solutions of the determining equations for the calculation of the elements of a \cinf-structure. In this regard, it would be very convenient to investigate strategies in order to simplify such determining equations, as it has been done in the literature for the search of $\mathcal{C}^{\infty}$-symmetries of some specific form.
Work in this line is currently in progress.

\section*{Acknowledgments}

The authors thank the finantial support from {\it   Junta de Andaluc\'ia }  to the research group FQM--377. A. Ruiz and C. Muriel acknowledge the financial support from {\it FEDER--Ministerio de Ciencia, Innovaci\'on y Universidades--Agencia Estatal
de Investigaci\'on} by means of the project PGC2018-101514-B-I00. 

\bibliographystyle{unsrt}
\bibliography{references}
\end{document}